\documentclass[conference]{IEEEtran}

\usepackage[dvips]{color}
\usepackage{comment}
\usepackage{todonotes}
\usepackage{epsf}
\usepackage{epsfig}
\usepackage{times}
\usepackage{epsfig}
\usepackage{graphicx}
\usepackage{bbold}
\usepackage{mathtools}
\usepackage{mathrsfs}
\usepackage{amssymb}
\usepackage[dvips]{color}
\usepackage{comment}
\usepackage{todonotes}
\usepackage{epsf}
\usepackage{epsfig}
\usepackage{times}
\usepackage{epsfig}
\usepackage{graphicx}
\usepackage{bbold}
\usepackage{mathtools}
\usepackage{mathrsfs}
\usepackage{amssymb}
\usepackage{pdfpages}
 \usepackage{ifpdf}
 
\usepackage[justification=centering]{caption}
\usepackage{subcaption}
\usepackage{dsfont}
\usepackage{lettrine} 
\usepackage{amsmath,epsfig,amssymb,algorithm,algpseudocode,amsthm,cite,url}
\usepackage{epstopdf}
\usepackage{caption}
\usepackage{subcaption}
\allowdisplaybreaks
\usepackage{csquotes}
\usepackage{verbatim}
\usepackage[english]{babel}
\usepackage{amsmath,amssymb}

\usepackage{verbatim} 
\usepackage{caption}
\newtheorem{theorem}{\bf Theorem}

\newtheorem{lemma}{\bf Lemma}

\usepackage{setspace}	
\newlength{\aligntop}
\newlength{\alignbot}
\makeatletter
\renewenvironment{align}{%
  \vspace{\aligntop}
  \start@align\@ne\st@rredfalse\m@ne
}{%
  \math@cr \black@\totwidth@
  \egroup
  \ifingather@
    \restorealignstate@
    \egroup
    \nonumber
    \ifnum0=`{\fi\iffalse}\fi
  \else
    $$%
  \fi
  \ignorespacesafterend%
  \vspace{\alignbot}\par\noindent
}

\IEEEoverridecommandlockouts
\begin{document}
\title{\huge Optimal Transport Theory for Power-Efficient\\ Deployment of Unmanned Aerial Vehicles \vspace{-0.4cm}}

\author{\authorblockN{ Mohammad Mozaffari$^1$, Walid Saad$^1$, Mehdi Bennis$^2$, and M\'erouane Debbah$^3$} \authorblockA{\small
$^1$ Wireless@VT, Electrical and Computer Engineering Department, Virginia Tech, VA, USA, Emails: \url{{mmozaff , walids}@vt.deu}\\
$^2$ CWC - Centre for Wireless Communications, Oulu, Finland, Email: \url{bennis@ee.oulu.fi}\\
$^3$ Mathematical and Algorithmic Sciences Lab, Huawei France R \& D, Paris, France, Email:\url{merouane.debbah@huawei.com}\\
 }%
 \thanks{This research was supported by the U.S. National Science Foundation under Grant AST-1506297.}
}
\maketitle

\begin{abstract}
In this paper, the optimal deployment of multiple unmanned  aerial vehicles (UAVs) acting as flying base stations is investigated. Considering the downlink scenario, the goal is to minimize the total required transmit power of UAVs while satisfying the users' rate requirements. To this end, the optimal locations of UAVs as well as the cell boundaries of their coverage areas are determined. To find those optimal parameters, the problem is divided into two sub-problems that are solved iteratively. In the first sub-problem, given the cell boundaries corresponding to each UAV, the optimal locations of the UAVs are derived using the facility location framework. In the second sub-problem, the locations of  UAVs are assumed to be fixed, and the optimal cell boundaries are obtained using tools from optimal transport theory. The analytical results show that the total required transmit power is significantly reduced by determining the optimal coverage areas for UAVs. These results also show that, moving the UAVs based on users' distribution, and adjusting their altitudes can lead to a minimum power consumption. Finally, it is shown that the proposed deployment approach, can improve the system's power efficiency by a factor of 20$\times$  compared to the classical Voronoi cell association technique with fixed UAVs locations.\vspace{-0.1cm}
\end{abstract}

\section{Introduction}\vspace{-0.01cm}
Unmanned aerial vehicles (UAVs) can be used as aerial base stations in order to satisfy the  coverage and rate requirements of wireless users \cite{Bucaille} and \cite{zhan}. Due the flying nature of UAVs, they can have line-of-sight (LOS) connections towards ground users thus leading to an improved coverage and rate performance. Compared to terrestrial base stations, mobile UAVs can intelligently move and change their location to provide on-demand coverage for ground users. As a result, UAV-based aerial base stations can be used to boost the wireless capacity and coverage at temporary events or hotspot areas. Effectively deploying UAVs requires meeting key challenges such as air-to-ground channel modeling, optimal deployment, path planning, and energy-efficient operation \cite{Uragun}.
  
The majority of the literature on UAV communications has mostly focused on empirical and analytical studies on air-to-ground channel modeling \cite{HouraniModeling, FengPath, Holis}. For example, in \cite{HouraniModeling}, the probability of LOS for air-to-ground communication was derived as a function of the elevation angle for different environments. The path loss model for the air-to-ground channel was further investigated in\cite{FengPath} and \cite{Holis}. Due to the existence of both path loss and shadowing, the characteristics of the air-to-ground channel are shown in \cite{Holis} to significantly depend  on the altitude of the aerial base stations.

To investigate UAV deployment, the work in \cite{HouraniOptimal} derives the optimal altitude for a single static UAV which yields a maximum coverage range. However, in this work, the authors did not consider the case of multiple UAVs nor the potential mobility of a UAV. In \cite{Mozaffari}, we investigated the impact of altitude on the minimum required transmit power for the case of two UAVs. However, our previous work is restricted to two UAVs and does not consider the impact of users' distribution on the deployment of UAVs.

Beyond deployment issues, the energy consumption of UAVs is also an important challenge. In fact, aerial base stations have a very limited amount of energy which must be used for transmission and mobility purposes \cite{Uragun}. In \cite{Zorbas}, the authors studied the energy efficiency of drones in target tracking scenarios by adjusting the number of active drones. However, in this work, the authors assumed that the location of targets is known in advance and they did not take into account the possible randomness of the network. The work in \cite{Ceran} investigated an optimal resource allocation scheme for an energy harvesting flying access point. However, the authors in \cite{Ceran} did not consider the case of multiple access points or UAVs. Furthermore, in their model, rate requirements are not guaranteed for all ground users.
Indeed, none of the previous studies \cite{Uragun,HouraniModeling} and \cite{Mozaffari,HouraniOptimal,Zorbas}, optimized the power of multiple, mobile UAVs by determining the optimal coverage regions of UAVs and their optimal locations while considering the ground users' geographical distribution. 

The main contribution of this paper is to develop a novel approach for optimally deploying UAVs to provide wireless to service ground users while minimizing the overall UAV transmit power needed to satisfy the users' data rate. We consider multiple UAVs in the downlink scenario and derive, jointly, the optimal cell boundaries (coverage area) and locations of UAVs that minimize the required transmit power. To this end, we first fix the cell boundaries, and solve the \emph{facility location problem}  \cite{Farahani} to determine the optimal locations of UAVs based on users' distribution. Next, given the prospective locations of UAVs, using \emph{optimal transport theory}, a powerful mathematical framework from probability theory \cite{Villani2008}, we find the optimal cell boundaries for the UAVs. Furthermore, we show that, there exist optimal altitudes at which the UAVs can satisfy the users' rate requirement while using minimum total transmit power. To our best knowledge, this is \emph{the first work that proposes a power-efficient deployment for UAVs by exploiting the optimal transport theory and facility location frameworks}. Here, we note that, optimal transport theory was used for resource allocation in traditional cellular networks in \cite{Silva} and \cite{ghazzai}. However, the results in these works do not extend to UAV scenarios and they do not deal with the efficient deployment problem. Compared to these studies, our analysis on  UAV-based aerial base stations is significantly different. First, we use different channel modeling between the UAVs and users based on probabilistic LOS links. In addition, unlike terrestrial base stations with fixed height, we consider an adjustable height for UAVs. Finally, we exploit the potential mobility of UAVs to improve power efficiency. 

The rest of this paper is organized as follows. Section II presents the system model describing the air-to-ground channel model. In Section III, the power minimization problem using transport theory and facility location problem is investigated. In Section IV, we present the numerical results. Finally, Section V concludes the paper. \vspace{-0.1cm}

\section{System Model}
Consider a geographical area divided into $K$ subareas in which $N$ wireless users are distributed based on an arbitrary distribution $f(x,y)$. This area must be serviced by multiple UAVs that will act as flying base stations as shown in Figure \ref{fig: UAV}. Each subarea will be served by a single UAV located at $({x_i},{y_i},{h_i})$ in the Cartesian coordinate where index $i$ corresponds to UAV $i$. Initially, we consider the subarea $i$ over $[{x_{s,i}},{x_{s,i + 1}}] \times [{y_{s,i}},{y_{s,i + 1}}] \subset \mathds{R}^2$.  We consider a downlink scenario in which UAVs adopt a frequency division multiple access (FDMA) technique to transmit data to the ground users at a desired data rate. FDMA assigns individual frequency bands to users and each user has its own dedicated channel for communications. Without loss of generality, we assume that the total transmit power of UAVs and the total available bandwidth is sufficient to meet the users' rate requirement. 

In our model, the UAVs transmit over different frequency bands and hence they do not interfere with one another. Moreover, hereinafter, we use the notion of a \emph{cell} to indicate the coverage region of each UAV. In other words, each UAV is associated with a cell within which the ground users serviced by this UAV are located. Note that, at the initial setup, the cell boundary associated with each UAV is not optimal, and our goal is to optimize those such that the total transmit power is minimized. Next, we first provide the air-to-ground channel model and, then, present the problem formulation.

\begin{figure}[!t]
  \begin{center}
   \vspace{-0.2cm}
    \includegraphics[width=8cm]{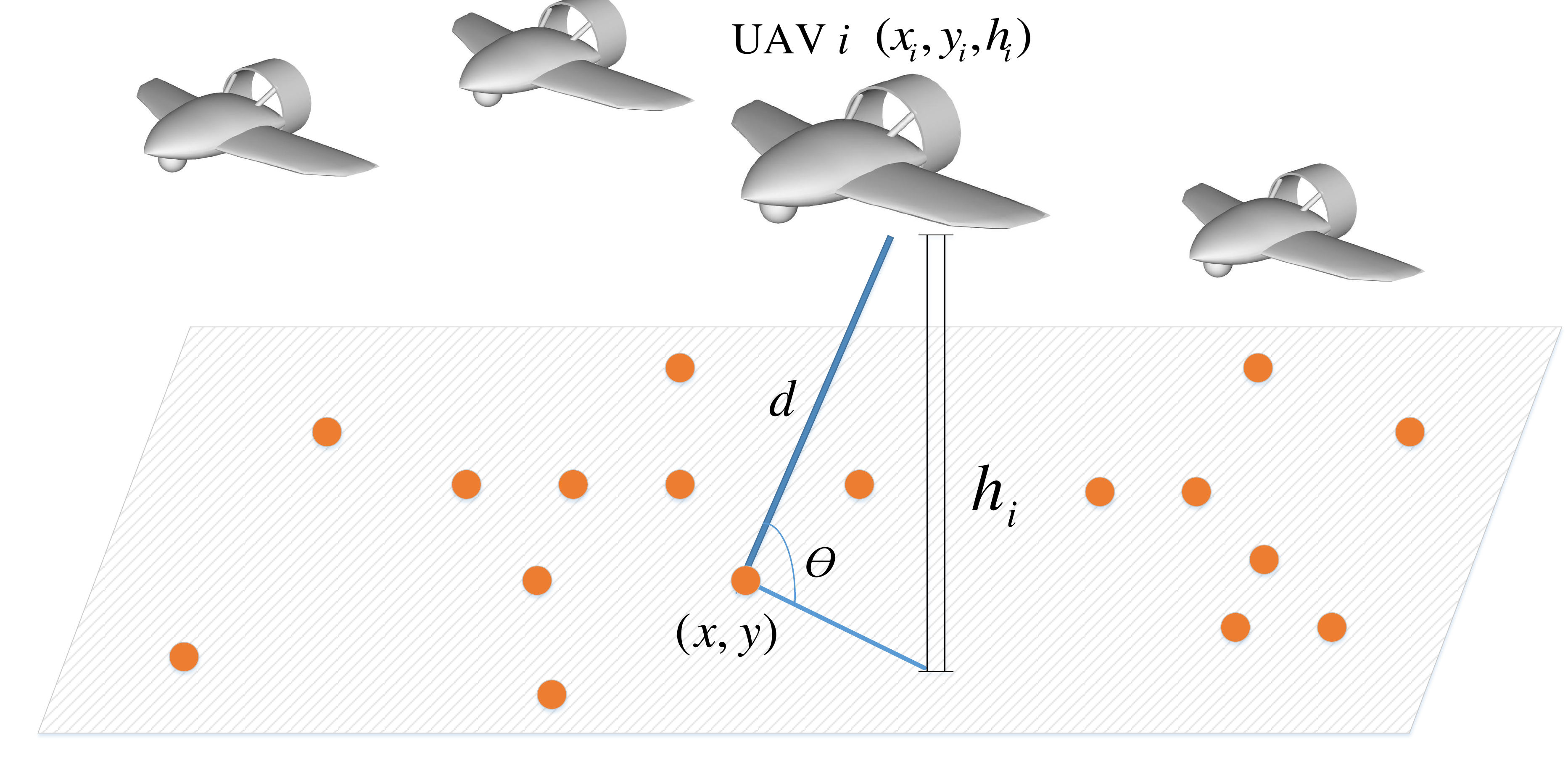}
    \vspace{-0.1cm}
    \caption{ UAVs transmitting data to ground users. \vspace{-0.5cm}}
    \label{fig: UAV}
  \end{center}\vspace{-0.2cm}
\end{figure}

\subsection{Air-to-ground channel model}
Typically, the air-to-ground propagation channel is modeled by considering the probabilistic LOS and non-line-of-sight (NLOS) links \cite{HouraniModeling}. In such a model, NLOS links experience higher attenuations due to the shadowing and diffraction loss. Therefore, in addition to the free space propagation loss, an additional path loss is considered for NLOS links.

The path loss model for LOS and NLOS links at the user location is given by \cite{HouraniModeling} and \cite{HouraniOptimal}:
\begin{equation}
{L_i(x,y)} = \left\{ \begin{array}{l}
{\left| {{{d_i}(x,y)}} \right|^{ {\alpha}}},{\rm{{\rm \hspace*{1.6cm}{LOS\hspace*{.2cm} link}}}},\\
\eta {\left| {{{d_i}(x,y)}} \right|^{  {\alpha }}},{\rm \hspace*{1.3cm} {   NLOS \hspace*{.2cm}link}},
\end{array} \right.
\end{equation}
where  $L$  is the path loss, $\left| {{{d_i}(x,y)}} \right|$  is the distance between a generic user located at at coordinates $(x,y)$ and UAV $i$, ${\alpha}$ is the path loss exponent over the user-UAV link, and  $\eta $ is an additional attenuation factor due to the NLOS connection. Here, the probability of LOS connection depends on the environment, density and height of buildings, the location of the user and the UAV, and the elevation angle between them. The LOS probability is given by  \cite{HouraniOptimal}:\vspace{-0.1cm}
\begin{equation}\label{PLOS}
{P_{{\rm{LOS}}}} = \frac{1}{{1 + C\exp ( - D\left[ {\theta  - C} \right])}},
\end{equation}
where $C$  and $D$  are constants which depend on the environment (rural, urban, dense urban, or others) and $\theta$  is the elevation angle.  Clearly, $\theta  = \frac{{180}}{\pi } \times {\sin ^{ - 1}}\left( {{\textstyle{{{h_i}} \over {{d_i}(x,y)}}}} \right)$, ${d_i}(x,y) = \sqrt {{{(x - {x_i})}^2} + {{(y - {y_i})}^2} + h_i^2}$ and also, the probability of NLOS is ${{P}_{{\text{NLOS}}}} = 1 - {{P}_{{\text{LOS}}}}$.\\
As observed from (\ref{PLOS}), the LOS probability increases as the elevation angle between the user and UAV increases.

Finally, the average path loss will be:
\begin{equation}\label{L_ave}
\bar L_i(x,y) = {P_{{\rm{LOS}}}}|{d_i}(x,y){|^\alpha } + \eta{P_{{\rm{NLOS}}}} |{d_i}(x,y){|^\alpha }.
\end{equation}

\subsection{Problem formulation}
Consider the transmission between UAV $i$ and a ground user located at $(x,y)$ coordinates. The achievable rate for the user is given by:\vspace{-0.2cm}
\begin{equation}\label{Rate}
R_i(x,y) = W_i{\log _2}\left( {1 + \frac{{{P_i}(x,y)/\bar L_i(x,y)}}{{{N_0}}}} \right),
\end{equation}
where $W_i$  is the transmission bandwidth of UAV $i$, ${P_i}(x,y)$  is the UAV transmit power to the user,  $\bar L_i(x,y)$ is the average path loss between UAV $i$ and the user, and ${N_0}$ is the noise power. Considering ${B_i}$ as the total available bandwidth at UAV $i$, and $M_i$ as the number of users serviced by UAV $i$, we have ${W_i} = \frac{{{B_i}}}{{{M_i}}}$. Note that, ${M_i}$ is the number of users inside cell boundary of UAV $i$ which is computed as ${M_i} = N\iint_{C_i} {f(x,y)} {\rm{ d}}x{\rm{d}}y$, with $N$ being the total number of users, and  ${C_i}$ being the cell boundary corresponding to UAV $i$. Clearly, the number of users covered by the UAV depends on the distribution of users, cell boundary, and the location of the UAV.

The minimum transmit power required to satisfy the rate requirement $\beta$ of ground users is given by: \vspace{-0.1cm}
\begin{equation}\label{P_min}
{P_{i,\min }}(x,y) = \left( {{2^{\beta /{W_i}}} - 1} \right){N_0}\bar L_i(x,y),
\end{equation}
which is derived using (\ref{Rate}) and $R(x,y) \ge \beta$.
Note that, as the number of users increases, the bandwidth per user decreases. Consequently, a higher transmit power is required to meet the users’ rate requirement. 

Given the location of UAVs, the average total transmit power of the UAVs in the network is given by:
\begin{equation}\label{P_t}
{P_t} = \sum\limits_{i = 1}^K {\iint_{C_i}{{P_{i,\min }}(x,y)f(x,y)}} {\rm{ d}}x{\rm{d}}y.
\end{equation}

Our goal is to minimize the total required transmit power by finding jointly the optimal locations of the UAVs and their associated cell boundaries. Therefore, the power minimization problem can be formulated as follows:
\begin{equation}\label{P_t_min}
\mathop {\min {P_t}}\limits_{C_i,x_i,y_i,h_i} = \sum\limits_{i = 1}^K {\iint_{C_i}{{\left( {{2^{\beta {M_i}/{B_i}}} - 1} \right){N_0}\bar L_i(x,y)}}} {\rm{ d}}x{\rm{d}}y{\rm{ }},
\end{equation}
where $i \in \{ 1,2,...,K\}$, $C_i$ is the cell boundary that shows the coverage region of the UAV, and $(x_i,y_i,h_i)$ is the location of UAV $i$ .

The solution for (\ref{P_t_min}) provides optimum cell boundaries and UAVs location such that the total average transmit power of UAVs is minimized while the rate requirement for all the users is maintained. However, solving (\ref{P_t_min}) is challenging due to the mutual dependency of $M_i$, $x_i$, $y_i$, $h_i$, and $C_i$. \textcolor{black}{ Furthermore, the problem must be solved over a continuous space while considering an infinite number of possible UAVs' locations and cell boundaries }. Therefore, there is a need for advanced mathematical tools that can help in finding the solution of (\ref{P_t_min}).\vspace{-0.01cm}
\section{Power optimization: Optimal transport theory and facility location}\label{sec:inter}
In order to solve (\ref{P_t_min}), we separate the problem into two optimization problems and solve them sequentially. In the first problem, we assume that the cell boundaries for UAVs are given and the objective is to determine the optimal location of UAVs for which the transmit power is minimized. In the second problem, for the given locations of UAVs, we derive the optimal cell association which lead to the minimum total required transmit power. \vspace{-0.15cm}
\subsection{Optimal UAVs location given cell boundaries}
Here, we first consider a scenario in which the UAVs  move and change their locations based on the users’ distribution. In order to minimize the total transmit power, given the cell boundaries, we find the optimal location of each UAV in its corresponding subarea. To this end, we use the \textit{facility location framework} \cite{Farahani}. In the facility location problem, given sets of facilities and clients, the goal is to find the optimal placement of facilities to minimize total transportation costs between the clients and facilities. Here, we consider the UAVs as the facilities, users as the clients, and the transmit power as the transportation costs. Consequently, given the distribution of users over the geographical area, the total transmit power of UAVs is given by ${P_t} = \sum\limits_{i = 1}^K {{P_{i}}}$.
where $P_{i}$ is the average transmit power of UAV $i$ over its corresponding subarea which is given by:\vspace{-0.1cm}
\begin{equation}\label{Pi_facility}
{P_{i}} = \int\limits_{{y_k}}^{{y_{k + 1}}} {\int\limits_{{x_k}}^{{x_{k + 1}}} {\left( {{2^{\beta /{W_i}}} - 1} \right){N_0}\bar L_i(x,y)} } f(x,y){\rm{ d}}x{\rm{d}}y.
\end{equation}\vspace{-0.4cm}

Note that, for a given cell boundary, the number of users inside the cell is fixed and optimizing $P_t$ and $P_i$ will not depend on $M_i$. Furthermore, in this case, minimizing $P_t$ is equivalent to minimizing $P_{i}$ for all  $i \in \{ 1,2,...,K\}$.
Therefore, the optimization problem in (\ref{P_t_min}) can be written as: \vspace{0.1cm}
 \begin{align}
\mathop {\min }\limits_{{x_i},{y_i},{h_i}} {P_{i}} = \int\limits_{{y_{s,i}}}^{{y_{s,i + 1}}} {\int\limits_{{x_{s,i}}}^{{x_{s,i + 1}}} {} } {\left( {{{(x - {x_i})}^2} + {{(y - {y_i})}^2} + h_i^2} \right)}\nonumber\\\times\left( {{P_{{\rm{LOS}}}} + (1 - {P_{{\rm{LOS}}}})\eta } \right)f(x,y){\rm{ d}}x{\rm{d}}y,
\end{align}
where we use $\alpha=2$ as the typical value of path loss exponent in LOS communications \cite{HouraniOptimal} .

Now, we derive a closed-form expression for the optimal location of the UAVs when they are deployed at high or low altitudes relative to the size of the subareas. \vspace{-0.2cm}
\begin{theorem} 
 \normalfont
Seeking a minimum required transmit power, the optimal location of  UAV $i$   positioned at  high or low altitudes compared to the size of its corresponding subarea ($\sqrt{{(x - {x_i})^2} + {(y - {y_i})^2}}$), is given by:
\begin{equation}\label{x_opt}
x_i^{{\rm{*}}} = \frac{{\int\limits_{{y_k}}^{{y_{k + 1}}} {\int\limits_{{x_k}}^{{x_{k + 1}}} x } f(x,y){\rm{ d}}x{\rm{d}}y}}{{\int\limits_{{y_k}}^{{y_{k + 1}}} {\int\limits_{{x_k}}^{{x_{k + 1}}} {f(x,y){\rm{ d}}x{\rm{d}}y} } }},
\end{equation}
\begin{equation} \label{y_opt}
y_i^{{\rm{*}}} = \frac{{\int\limits_{{y_k}}^{{y_{k + 1}}} {\int\limits_{{x_k}}^{{x_{k + 1}}} y } f(x,y){\rm{ d}}x{\rm{d}}y}}{{\int\limits_{{y_k}}^{{y_{k + 1}}} {\int\limits_{{x_k}}^{{x_{k + 1}}} {f(x,y){\rm{ d}}x{\rm{d}}y} } }},
\end{equation}
where $h_i^2 >> {(x - {x_i})^2} + {(y - {y_i})^2}$ or $h_i^2 << {(x - {x_i})^2} + {(y - {y_i})^2}$. 
Note that, considering $f(x,y)$  as the distribution of users over the geographical area, $(x_i^{{\rm{*}}},y_i^{{\rm{*}}})$ corresponds to the centroid the area.  
\end{theorem}

\begin{proof}
At very high altitudes, ${d_i}(x,y) \approx {h_i} \to \theta  = \frac{{180}}{\pi } \times {\sin ^{ - 1}}\left( {{\textstyle{h \over {{d_i}(x,y)}}}} \right) \approx {90^o}$, and hence, ${P_{{\rm{LOS}}}} \approx 1$.
Then we have, $\bar L_i(x,y) \approx \left( {{{(x - {x_i})}^2} + {{(y - {y_i})}^2} + h_i^2} \right)$, $\frac{{\partial \bar L_i(x,y)}}{{\partial {x_i}}} = 2(x - {x_i})$, and $\frac{{\partial \bar L_i(x,y)}}{{\partial {y_i}}} = 2(y - {y_i})$.\\
At very low altitudes, ${P_{{\rm{LOS}}}} \approx 1$, $\bar L_i(x,y) = \eta \left( {{{(x - {x_i})}^2} + {{(y - {y_i})}^2} + h_i^2} \right)$, $\frac{{\partial \bar L_i(x,y)}}{{\partial {x_i}}} = 2\eta (x - {x_i})$, and $\frac{{\partial \bar L_i(x,y)}}{{\partial {y_i}}} = 2\eta (y - {y_i})$. As a result,

\begin{align} \label{x_o}
\frac{{\partial {P_{i}}}}{{\partial {x_i}}} &= \frac{{\partial \int\limits_{{y_k}}^{{y_{k + 1}}} {\int\limits_{{x_k}}^{{x_{k + 1}}} {\bar L(x,y)} } f(x,y){\rm{ d}}x{\rm{d}}y}}{{\partial {x_i}}} \nonumber\\
&=\int\limits_{{y_k}}^{{y_{k + 1}}} {\int\limits_{{x_k}}^{{x_{k + 1}}} {\frac{{\partial \bar L(x,y)}}{{\partial {x_i}}}} } f(x,y){\rm{ d}}x{\rm{d}}y = 0.
\end{align}
 
 Likewise,
 \begin{align} \label{y_o}
 \frac{{\partial {P_{i}}}}{{\partial {y_i}}} = \int\limits_{{y_k}}^{{y_{k + 1}}} {\int\limits_{{x_k}}^{{x_{k + 1}}} {\frac{{\partial \bar L(x,y)}}{{\partial {y_i}}}} } f(x,y){\rm{ d}}x{\rm{d}}y = 0.
 \end{align} \vspace{0.2cm}
 Finally, solving (\ref{x_o}) and (\ref{y_o}) leads to (\ref{x_opt}) and (\ref{y_opt}).
\end{proof}

Next, we provide an approximation for the optimal location of UAV $i$ deployed at an arbitrary altitude $h_i$
\begin{theorem} 
 \normalfont
An approximation for the optimal location of UAV $i$  is the solution of the following system of equations:
\begin{equation} \label{System_eq}
\left\{ \begin{array}{l}
{g_1}({x_i},{y_i}) = {a_1}x_i^3 + {a_2}x_i^2 + {a_3}{x_i} + {a_4} = 0,\\
{g_2}({x_i},{y_i}) = {b_1}y_i^3 + {b_2}y_i^2 + {b_3}{y_i} + {b_4} = 0,
\end{array} \right.
\end{equation}
where $a_1,...,a_4$ and $b_1,...,b_4$ coefficients will be presented in the proof.    
\end{theorem} 

\begin{proof} 
In order to find the optimal location of UAV $i$ $({x_i},{y_i})$, we first provide an approximation for ${P_{{\rm{LOS}}}}$  that simplifies our analysis. Using two steps least square, ${P_{{\rm{LOS}}}}$ for $100\,{\rm{ m}} \le {h_i} \le 2000\,{\rm{ m}}$  and $\sqrt {{{(x - {x_i})}^2} + {{(y - {y_i})}^2}}  \le 1000\,{\rm{m }}$ can be approximated by:
\begin{equation}
{P_{{\rm{LOS}}}} \approx {f_1}({h_i})\left( {{{(x - {x_i})}^2} + {{(y - {y_i})}^2}} \right) + {f_2}({h_i}),
\end{equation}
where  ${f_1}({h_i}) = - {10^{ - 11}}h_i^2 + 15 \times {10^{ - 9}}{h_i} - 57 \times {10^{ - 7}}$, and ${f_2}({h_i}) = 2.37 \times {10^{ - 7}}h_i^2 - {\rm{5}}{\rm{.24}} \times {10^{ - 4}}{h_i} + 1.32$.

According to (\ref{L_ave}), (\ref{x_o}), (\ref{y_o}), and considering the fact that $\bar L(x,y)$  is continuous and differentiable, the system of equations presented in (\ref{System_eq}) is derived with
 ${a_1} = 4\lambda {F_{{\mathop{\rm int}} }}(1)$, ${a_2} =  - 12\lambda {F_{{\mathop{\rm int}} }}(x)$, ${a_3} = {F_{{\mathop{\rm int}} }}\left( {2q + 4qh_i^2 + 4\lambda x\left[ {3x + {{(y - {y_i})}^2}} \right]} \right)$, ${a_4} =  - {F_{{\mathop{\rm int}} }}\left( {2qx + 4\lambda x\left[ {{x^2} + h_i^2 + {{(y - {y_i})}^2}} \right]} \right)$, and ${F_{{\mathop{\rm int}} }}\left( {u(x,y)} \right) = \int\limits_{{y_{s,i}}}^{{y_{s,i + 1}}} {\int\limits_{{x_{s,i}}}^{{x_{s,i + 1}}} {u(x,y)f(x,y){\rm{ d}}x{\rm{d}}y} }$. Also, we have $q = \eta  + (1 - \eta ){f_2}({h_i})$, and $\lambda  = (1 - \eta ){f_1}({h_i})$.
Note that, in (\ref{System_eq}), ${g_2}({x_i},{y_i})$  can be found from ${g_1}({x_i},{y_i})$  by substituting ${y_i}$  instead of ${x_i}$, and $y$ instead of $x$ except for $f(x,y)$. 

 This system of equations can be solved using the Newton-Raphson method \cite{ben}. \textcolor{black}{Note that, since $P_i$ is continuous and differentiable, $x_{s,i}<x_i<x_{s,i+1}$, and $y_{s,i}<y_i<y_{s,i+1}$, an optimal $(x,y)$ which leads to  a minimum $P_i$ exists. Therefore, (\ref{System_eq}) has a solution.}   
\end{proof} 
Based on Theorems 1 and 2, we can find the optimal locations of UAVs over the given cell boundaries. Next, we derive the optimal cell boundaries assuming that the locations of the UAVs are fixed.\vspace{-0.35cm}
\subsection{Optimal cell boundaries given UAVs location}\vspace{-0.1cm}
In this optimization problem, the goal is to derive the optimal cell boundaries for all UAVs, assuming that the locations of the UAVs are given. The optimal cell association for which the total transmit power is minimized, depends on the distribution of the users. 

Given the locations of UAVs, the optimization problem in (\ref{P_t_min}) can be expressed by:
\begin{equation}\label{P_t_transport}
\mathop {\min {P_t}}\limits_{C_i} = \sum\limits_{i = 1}^K {\iint_{C_i}{{\left( {{2^{\beta {M_i}/{B_i}}} - 1} \right){N_0}\bar L_i(x,y)}}} {\rm{ d}}x{\rm{d}}y{\rm{ }}.
\end{equation}

Note that, in (\ref{P_t_transport}), the cell boundaries impact the number of users located inside the cells, and consequently, the total transmit power. Furthermore, according to (\ref{P_t_transport}),  the transmit power of each UAV exponentially increases as a function of the number of users that it serves. Hence, in the cell association problem, in addition to the channel gain, $\bar L(x,y)$, the number of users in the cell should also be taken into account.  

To solve (\ref{P_t_min}), we draw a striking analogy with the framework of \emph{optimal transport theory}\cite{Villani2008}. Optimal transport theory is a mathematical framework for studying the optimal transportation map between two probability distributions. Optimal transport theory was first introduced by Monge while dealing with a transportation problem. In the Monge problem, piles of sand and some holes with the same total volume of sands are randomly distributed over an geographical area. The objective is to find the optimal moves to transport the entire piles to the holes with a minimum transportation cost. Intuitively, the optimal moves depend on the cost function which is a function of distance between pills and holes. Our optimization problem in (\ref{P_t_transport}), can be considered as finding optimal moves to transport data from UAVs to the location of users. In this case, the cost of transportation is the minimum required transmit power that satisfies the users' rate requirement. Furthermore, in our model, the best moves corresponds to determining the best users associated with each UAV which leads to the optimal cell boundaries. Now, we rewrite (\ref{P_t_transport}) as follows:
\begin{align}\label{P_t_transport2}
\mathop {\min }\limits_{{C_i}{\rm{ }}} \sum\limits_{i = 1}^K {\iint_{{C_i}} {F\left( {{d_i}(x,y)} \right)} }. S({M_i}) f(x,y){\rm{d}}x{\rm{d}}y,
\end{align}
where ${M_i} = \iint_{C_i}{{f(u,v){\rm{ d}}u{\rm{d}}v}}$, $S({M_i}) = \left( {{2^{\beta {M_i}/{B_i}}} - 1} \right)$, and
\begin{equation}\label{F}
F({d_i}(x,y)) = {N_0}\bar L_i(x,y).
\end{equation}
 
To solve (\ref{P_t_transport2}), we invoke a lemma from optimal transport theory \cite{Silva}. \vspace{-0.1cm}
\begin {lemma}
 \normalfont
Assuming that $F$ is a continuous function, and $S$ is derivable, the solution for (\ref{P_t_transport2}) is given by \cite{Silva}:
\begin{equation}\label{C_i}
{C_i} = \left\{\begin{array}{l}
\hspace{-0.2cm}(x,y);S({M_i})F({d_i}(x,y))f(x,y) + {T_i}(x,y) \le\\
 S({M_j})F({d_j}(x,y))f(x,y) + {T_j}(x,y),\forall i \ne j
\end{array}\hspace{-0.2cm}\right\},
\end{equation}
where ${M_i} = \iint_{C_i}{{f(u,v){\rm{ d}}u{\rm{d}}v}}$, and ${T_i}(x,y)=\frac{{\partial S}}{{\partial {M_i}}}\iint_{C_i}{{F({d_i}(x,y))f(x,y){\rm{  d}}x{\rm{d}}y}}$.
\end{lemma}

As we can see, the solution in (\ref{C_i}) provides optimal cell boundaries, $(x,y)$, for each UAV, in such a way that total transmit power is minimized. \textcolor{black}{Given Lemma 1, to solve (\ref{P_t_transport2}), a total of ${K \choose 2}$ non-linear equations must be solved.} Note that, the optimal cell boundaries, (\ref{C_i}), are  a function of users’ distribution. In the provided solution, there is \emph{no specific assumption on the distribution of users and it holds for any arbitrary distribution}. Hereinafter, we consider  the uniform and \textcolor{black}{truncated Gaussian distributions}  for users location. These distributions considering a geographical area with size of   ${L_x} \times {L_y}$ are given by:
\begin{equation}
f(x,y) = \frac{1}{{{L_x}{L_y}}},\hspace{.3cm}{\textnormal{   for uniform distribution}},
\end{equation} \vspace{-0.4cm}
\begin{equation}
f(x,y) = \frac{1}{G}{\rm{exp}}{\left( {\frac{{{L_x} - {\mu _x}}}{{\sqrt {2{\sigma _x}} }}} \right)^2}{\rm{exp}}{\left( {\frac{{{L_y} - {\mu _y}}}{{\sqrt {2{\sigma _y}} }}} \right)^2},
\end{equation}
where $G = 2\pi {\sigma _x}{\sigma _y}{\rm{erf}}\left( {\frac{{{L_x} - {\mu _x}}}{{\sqrt {2{\sigma _x}} }}} \right){\rm{erf}}\left( {\frac{{{L_y} - {\mu _y}}}{{\sqrt {2{\sigma _y}} }}} \right)$, ${\mu _x}$, ${\sigma _x}$, ${\mu _y}$, and ${\sigma _y}$ are the mean and standard deviation values in the $x$ and $y$ directions, and ${\rm{erf(}}z{\rm{)}} = \frac{2}{{\sqrt \pi  }}\int\limits_0^z {{e^{ - {t^2}}}{\rm{d}}t}$.

The \textcolor{black}{truncated Gaussian distributions} is used for modeling a hotspot area in which users are highly distributed around the hotspot center and their density decreases at locations further away from the center. The center of hotspot is denoted by (${\mu _x}$, ${\mu _y}$), and the density around the center depends on the ${\sigma _x}$  and ${\sigma _y}$  values. Here, we define ${\rho _x} = \frac{1}{{{\sigma _x}}}$  and  ${\rho _y} = \frac{1}{{{\sigma _y}}}$   as the density of user in $x$ and $y$ directions respectively. 

In summary, to achieve the optimal solution for (\ref{P_t_min}), first we generate an initial cell boundary for each UAV. Next, following the discussion in Subsection III.A, we position each UAV at the optimal location over its coverage region. Given the new UAVs locations, we determine a new set of cell boundaries by following Subsection III.B. In the next iteration, the optimal location of UAVs for the updated cell boundaries are obtained, and then, the cell boundaries are updated again based on the latest UAVs locations. This process continues until the optimal solution is reached. \vspace{-0.2cm}
\section{Numerical Results}\label{sec:Results}\vspace{-0.1cm}

We consider a rectangular $1000$\,m$\times 500$\,m area. The area is divided into two equal subareas and contains two UAVs. We also consider a hotspot area in which users are distributed according to a \textcolor{black}{truncated Gaussian distributions} around a center with $({\mu _x} =  - 100{\rm{\,m }},{\mu _y} = 100{\rm{\,m }})$, and a density ${\rho _x} = {\rho _y} = \rho$. Hereafter, we denote the UAV closer to the hotspot center by UAV$_{1}$, and the second UAV by UAV$_{2}$. In our analytical analysis, we set $N = 200$, $\beta  = 1$\,Mbps, and ${B_i} = 50$\,MHz. Moreover, we consider a dense urban environment in which $C=11.9$, $D=0.13$, and $\eta=100$ \cite{HouraniModeling}. 

Figure \ref{fig: P_density} shows the average transmit power of UAVs versus the density of the users for the optimal cell boundaries and the Voronoi cell boundaries. Note that, in Figure \ref{fig: P_density}, we assume that the UAVs are located at the center of the subareas\textcolor{black}{, and their altitude is 200\,m.} As we can see, the average transmit power for optimal cell boundaries is significantly lower than the Voronoi case. According to Figure \ref{fig: P_density}, the average transmit power is around 0.3\,W and 0.12\,W, respectively, for the Voronoi and the proposed optimal cell boundary cases. Furthermore, the Voronoi case is more sensitive to the users' density compared to the optimal cell boundaries. This is due to the fact that the optimal cell boundaries are determined based on the users' density such that the transmit power is minimized. However, in the Voronoi case, the cell boundaries are set without considering the users’ density. As observed in Figure  \ref{fig: P_density}, for the low user density case in which the users are more spread over the area, the performance of Voronoi and optimal cell boundaries are close. However, as the density increases, the proposed optimal case becomes better but then they get close again. The reason is that, for a very highly dense scenario, most users are located around the hotspot center and they are served  by the closest UAV. As a result, the average channel gain is high for the users and thus, power efficiency for the Voronoi case is improved. As we see from Figure \ref{fig: P_density}, for $\rho  \approx 0.02$,  the proposed approach yields a maximum power improvement over the Voronoi case. 

\begin{figure}[!t]
  \begin{center}
   \vspace{-0.2cm}
    \includegraphics[width=8.5cm]{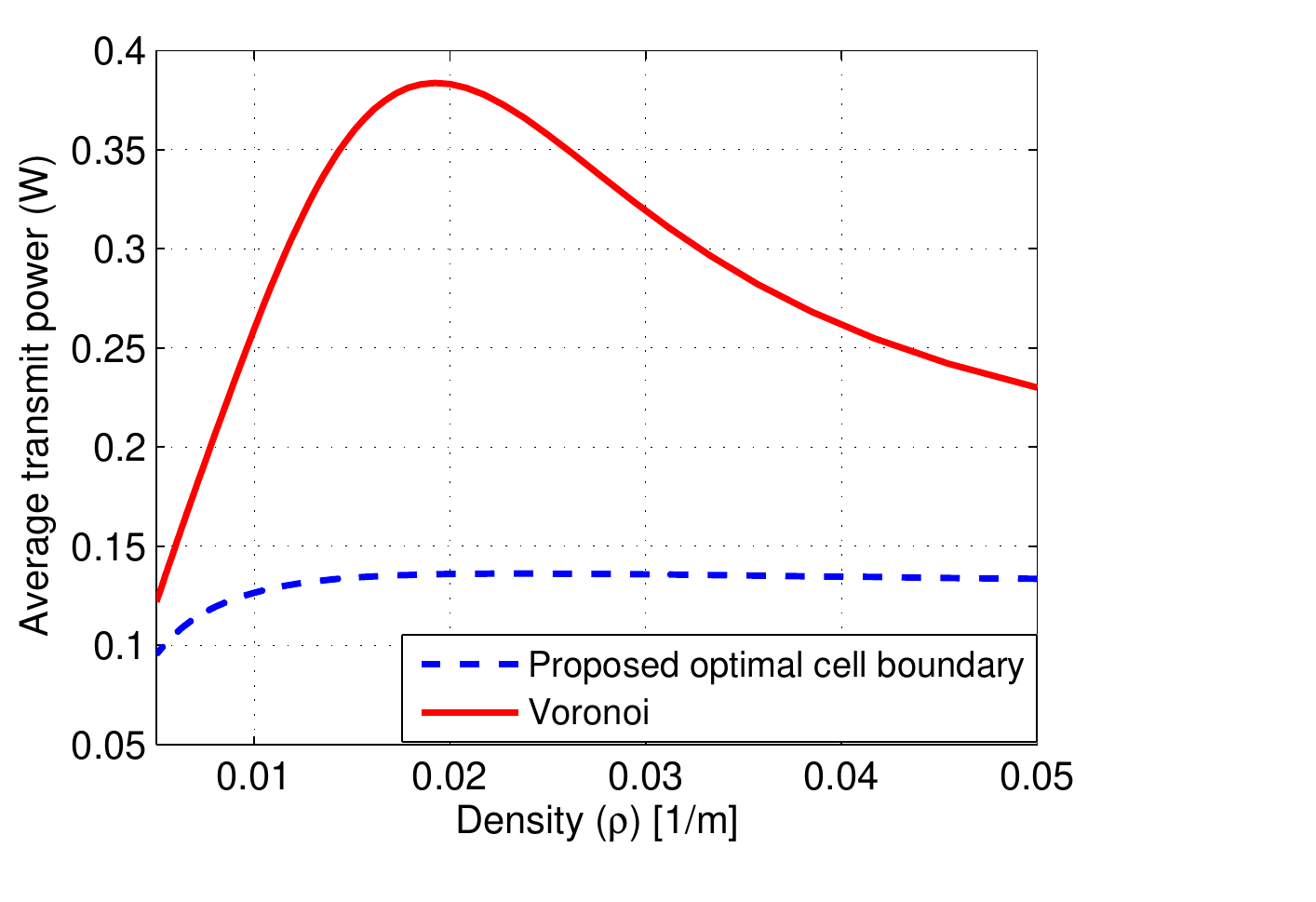}
    \vspace{-0.55cm}
    \caption{  \small Average required transmit power versus users' density. \vspace{-.2cm}}
    \label{fig: P_density}
  \end{center}\vspace{-0.2cm}
\end{figure}

Figure \ref{fig: P_H} shows the impact of UAV altitude on the average transmit power for optimal cell boundaries \textcolor{black}{with $\rho=0.01$.} In our setup, UAV$_1$ is closer to the hotspot center than UAV$_2$. Figure \ref{fig: P_H} shows that, the total average transmit power is minimum at an altitude of 400\,m. In fact, the UAVs should not be positioned at very low altitudes, due to high shadowing and a low probability of LOS connections towards the users. On the other hand, at very high altitudes, LOS links exist with a high probability but the large distance between UAV and users results in a high path loss. As shown in Figure \ref{fig: P_H}, the optimal individual altitude for UAV$_1$ and UAV$_2$ are around 320\,m and 500\,m respectively. However, the total transmit power of both UAVs is minimized for $h_1=h_2=400$\,m .

\begin{figure}[!t]
  \begin{center}
   \vspace{-0.2cm}
    \includegraphics[width=8.3cm]{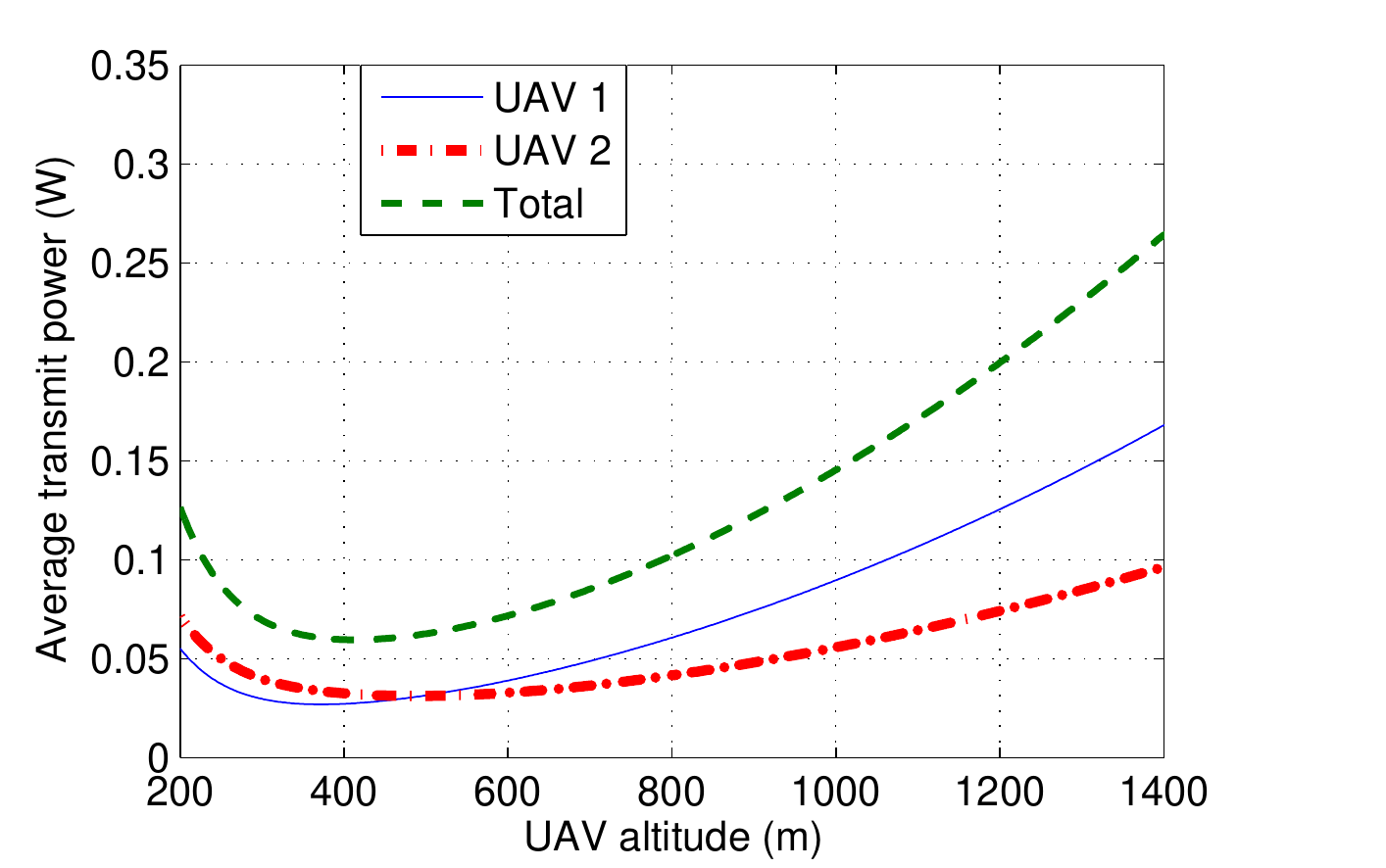}
    \vspace{-0.3cm}
    \caption{  \small Average required transmit power versus UAVs altitude. \vspace{-0.3cm}}
    \label{fig: P_H}
  \end{center}\vspace{-0.5cm}
\end{figure}

Figure \ref{fig: P_H12} illustrates the inverse of the total average transmit power as a function of altitude for the optimal cell association case. Note that, in this figure, we used inverse of power solely for a better illustration of the results and for clarity of the figure. Here, we consider all the possible combinations of  and  from 200\,m to 1200\,m. As seen from Figure \ref{fig: P_H12}, the minimum total average transmit power (maximum inverse of power) is about 0.12\,W and it is achieved for $h_1=310$\,m and $h_2=530$\,m. Note that, since the hotspot center is closer to UAV$_1$, on the average, this UAV has a higher chance of LOS links to users compared to UAV$_2$. Hence, UAV$_2$ should be at a higher altitude in order to improve its channel condition (more LOS links) to the users.   

\begin{figure}[!t]
  \begin{center}
   \vspace{-0.2cm}
    \includegraphics[width=7.5cm]{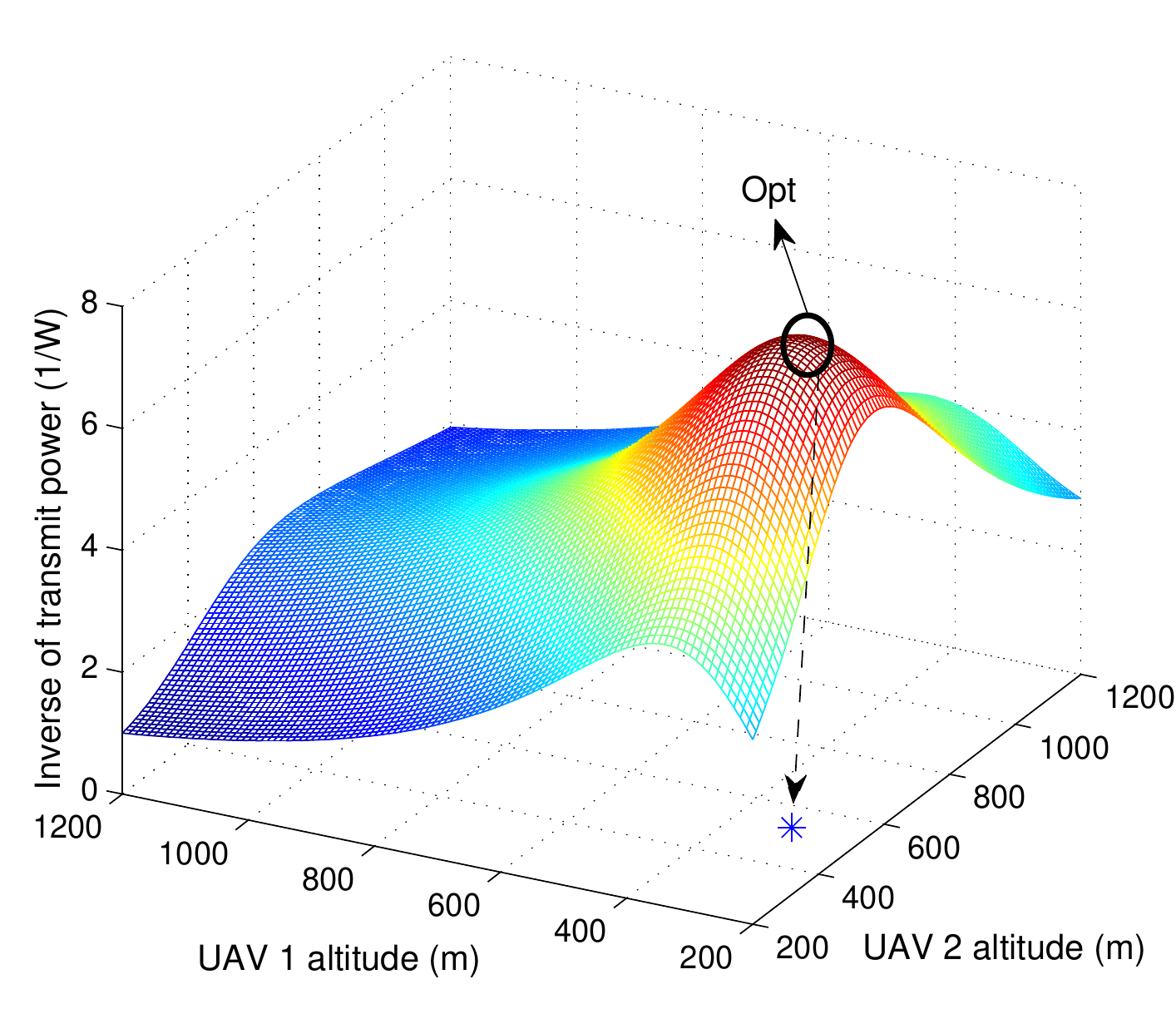}
    \vspace{-0.3cm}
    \caption{\small Inverse of average transmit power versus altitude. \vspace{-0.35cm}}
    \label{fig: P_H12}
  \end{center}\vspace{-0.1cm}
\end{figure}

%

Figure \ref{fig:P_optimal_comparison}  shows the average required transmit power for optimal cell association, optimal UAVs' location, and combined optimal cell association and UAVs location cases. As expected, the combined cell-association/UAV-location outperforms the UAV-location and cell association cases by factor of 3 and 10 respectively. Moreover, considering Figure \ref{fig: P_density}, the proposed combined method improves the power efficiency by factor of 20 compared to the Voronoi case. In the combined case, the mobile UAVs move to their optimal locations based on the users' distribution. As a result, the UAVs need  lower transmit power to satisfy the users’ rate requirement. Moreover, once the UAVs reach the new location, the optimal cell boundaries are updated and further improvement in the power efficiency is achieved.\vspace{-0.3cm}

\begin{figure}[!t]
  \begin{center}
   \vspace{0.012cm}
    \includegraphics[width=8cm]{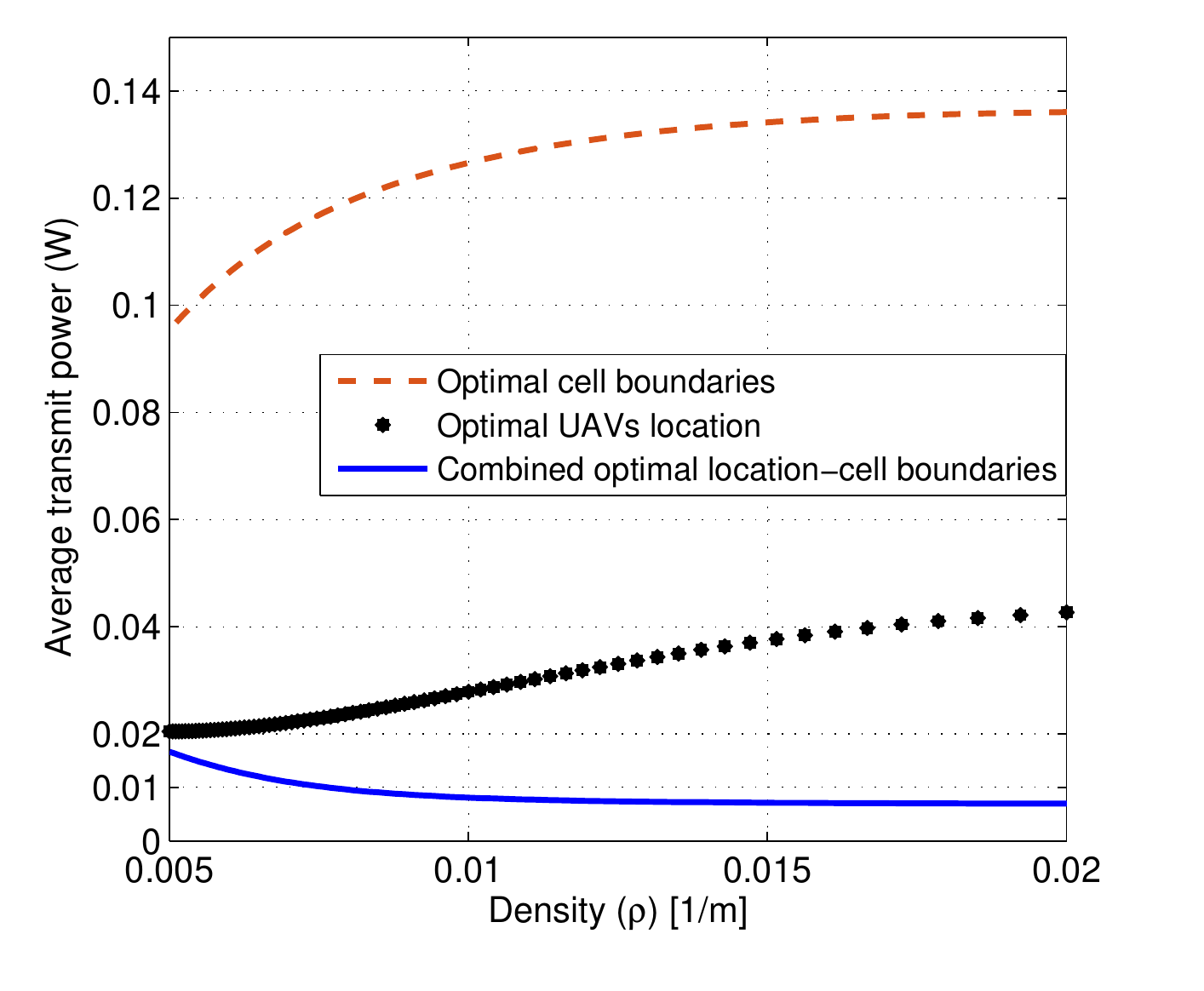}
    \vspace{-0.5cm}
    \caption{ \small  Average required transmit power versus users density. \vspace{-0.1cm}}
    \label{fig:P_optimal_comparison}
  \end{center}\vspace{-0.6cm}
\end{figure}


\section{Conclusions}\label{sec:Results}\vspace{-0.15cm}

In this paper, we have proposed a novel, optimal deployment framework for deploying UAVs that act as flying base stations. We have cast the problem as a power minimization problem under the constraint of satisfying the rate requirement for all ground users. To this end, we have applied optimal transport theory to obtain the optimal cell association, and we have derived optimal UAVs' locations using the facility location framework. Moreover, we have investigated the impact of UAVs altitude on the power efficiency. The results have shown that the total required transmit power is significantly decreased by determining the optimal coverage regions for UAVs. Furthermore, finding the optimal UAVs locations based on the users' distribution, and adjusting their altitudes to optimal values will yield the minimum power consumption. \vspace{-0.4cm}

\def\baselinestretch{1.03}
\bibliographystyle{IEEEtran}
\bibliography{references}
\end{document}